\documentclass[12pt,a4paper]{article}
\usepackage[top=1in, bottom=1.25in, left=1.2in, right=1.2in]{geometry}
\usepackage[ruled]{algorithm2e}
\usepackage{graphicx}
\usepackage{subcaption}
\usepackage{amsthm}
\usepackage{amsmath}
\usepackage{amsfonts}
\usepackage{amssymb}

\newtheorem{theorem}{Theorem}

\title{A Graph-Based Platform for Customer Behavior Analysis using Applications' Clickstream Data}

\author{Mojgan Mohajer\\
	BMW AG\\
	\texttt{mojgan.mohajer@bmw.de}}

\begin{document}

\maketitle

\centerline{Technical Report}

\begin{abstract}
Clickstream analysis is getting more attention since the increase of usage in e-commerce and applications. Beside customers' purchase behavior analysis, there is also attempt to analyze the customer behavior in relation to the quality of web or application design. In general, clickstream data can be considered as a sequence of log events collected at different levels of web/app usage. The analysis of clickstream data can be performed directly as sequence analysis or by extracting features from sequences. In this work, we show how representing and saving the sequences with their underlying graph structures can induce a platform for customer behavior analysis. Our main idea is that clickstream data containing sequences of actions of an application, are walks of the corresponding finite state automaton (FSA) of that application. Our hypothesis is that the customers of an application normally do not use all possible walks through that FSA and the number of actual walks is much smaller than total number of possible walks through the FSA. Sequences of such a walk normally consist of a finite number of cycles on FSA graphs. Identifying and matching these cycles in the classical sequence analysis is not straight forward. We show that representing the sequences through their underlying graph structures not only groups the sequences automatically but also provides a compressed data representation of the original sequences.
\end{abstract}


\section{Introduction}
Collecting and analyzing clickstream data is growing since the increase of usage in e-commerce and applications. The focus is mainly of on purchasing and prediction of customers' purchase behavior~\cite{srivastava2000web, banerjee2001clickstream, moe2002bayesian, heer2002separating, montgomery2004modeling, baumann2018changing, wang2017clickstream}. However, there have also been attempts to analyze the customer behavior with respect to the usability of web or application design~\cite{bhatti2000integrating,ting2005ubb, wang2013you, geng2014improving}. In general, clickstream data can be considered as a sequence of log events collected at different levels of web/app usage, such as client or server level. It can be also considered at different levels of granularity, for example, page visits or action clicks on each visited page. The analysis of clickstream data has developed in several directions during the previous years. Some of the works focus on feature extraction from clickstream sequences for further analysis~\cite{bhatti2000integrating,srivastava2000web,moe2002bayesian, heer2002separating}. There are some other works that use clickstream sequences as input data and perform sequence analysis methods directly on sequences~\cite{banerjee2001clickstream,montgomery2004modeling,lu2005mining,ting2005ubb,wang2013you,wang2016unsupervised,baumann2018changing,wang2017clickstream}. Path analysis is also related to this topic, where a specific path is defined and customer behavior is analyzed through this path to identify possible problems and bottle necks~\cite{geng2014improving, mah2006method}. 
For analyzing the customer behavior towards quality analysis of application design, extracting features or doing path analysis can be meaningful. However, such methods suffer from the lack of generality. In both methods, we define explicitly our focus of analysis from the beginning. On the other hand, generalized methods which analyze the whole clickstream sequences and try to find patterns to understand customer behavior, provide insights into the design problems, which are not clear from the beginning. In the works which consider the clickstream data as sequences, either perform classical sequence mining to find out, e.g., which pages were normally visited together~\cite{hay2004mining,pitman2010insights}, or cluster the users by defining some similarity measures between the sequences~\cite{banerjee2001clickstream,lu2005mining,ting2005ubb,wang2013you,wang2017clickstream}. These methods normally suffer from scalability and performance.

Using graphs for sequence analysis in this domain is mostly limited to similarity-graphs and corresponding algorithms on similarity-graphs~\cite{wang2016unsupervised,wang2017clickstream}. In a similarity-graph, each node represents a sequence and edges represent the distance between nodes. In this sense, the graph is not used to take the structure of the sequence into account but is a tool for performing specific clustering algorithms optimized for similarity-graphs. Yet, representing the sequence in a graph structure has some benefits. In biology, for example, they use graphs to represent parts of sequences, due to the fact that graph structure is invariant under some transformations such as mirroring~\cite{corel2016network,paten2017genome,kavya2019sequence, zhang2018network}. In the domain of clickstream analysis, Baumann et al.~\cite{baumann2018changing} use graph structure for modeling the clickstream sequences. They generate a graph for each sequence of click data and calculate graph metrics for that generated graph. Finally, they use the calculated graph metrics for further methods such as gradient boosted trees. They examine experimentally the importance of different metrics and showed, for example, that number of cycles and self-loops are significant metrics in prediction of customer purchasing behavior. But they loose the track of how many times a cycle was traversed.

A graph structure can capture the circles in a sequence and therefore provides a more compact representation of a sequence by counting the number of times a cycle was visited. It is specially the case if the sequences are the walks through a state automaton, which is the case in most of clickstream data from applications. Furthermore, if a sequence can be considered as a regular expression, then it can be reduced to a walk on a finite state automata as well~\cite{hingston2002using}. 

In this work, we show how representing and saving the sequences with their underlying graph components, can induce a platform for customer behavior analysis. The main idea behind our work is that clickstream data produced from actions in an application are walks of the corresponding finite state automate (FSA) of that application. Each action corresponds to a state of the FSA. Our hypothesis is that the customers of an application normally are not using all the possible walks through that FSA but the number of real walks is much smaller than the number of possible walks through the FSA. The second point is that the walks (sequences) consist of finite numbers of cycles on FSA graphs. Thus, these walks can be reduced to their common graph components. These components capture information which is not completely covered in the common used methods such as \textit{k-grams pattern analysis} in common sequence analysis. In \textit{k-grams} method, normally the number of common $k$-long sub-sequences of two sequences is calculated. The correct choice of $k$ is crucial. The multiple attendance of a cycle is captured somehow by counting the \textit{k-grams} but the meaning of the cycle itself is totally lost. 

We show that representing the sequences through their underlying graph structure not only groups the sequences automatically, but also is a compressed data representation of the original sequences. This compressed representation, allows queries directly on the saved data in graph database. In the original sequence format, for each of these queries, an specific algorithm is required. In the following sections, we describe the data and formal definitions used in this paper. Afterward, we introduce our method how to convert sequences to their graph components. At the end, we discuss the possible designs in a native graph database and introduce specific queries which can be performed better in this concept.

\section{Data}\label{data}
The data used for our analysis are log events of clickstream from application installed on board of a vehicle such as navigation, radio etc. It consists of sessions which can roughly be considered as a drive (it is still possible to have several sessions during a drive, depending on how a new id for a session is triggered). During a drive the actions of the user on the HMI (Human-Machine Interface) are collected by a event-log mechanism. These events correspond to the states of a FSA. We consider the consequent occurrence of the states of an specific application during a session as a sequence of the clickstream for that application and user. The states in a sequence also have a timestamp which defines the order of the states in a sequence. A vehicle (user) can have several drives (sessions) over time. For each session, exists a sequence for every running application. Each sequence will be represented in its corresponding graph components as described in the next section.

\section{Sequence to Graph Structure}\label{comps}
\subsection{Definitions}
We define a graph $G \left\langle V,E\right\rangle$  with a set of vertices $V$ and edges $E$ as a directed graph, such that an edge $ e \in E$ is a tuple $(u,v)$ with starting vertex $u$ and ending vertex $v$. 
\begin{itemize}
	\item  A \textit{(directed) walk} in such a directed graph $G$ is a finite sequence of edges $(e_1, e_2, …, e_{n-1})$ which joins a sequence of vertices $(v_1, v_2, …, v_n)$ such that $e_i = (v_i, v_{i+1})$ for $i = 1, 2, … , (n-1) $ and $(v_1, v_2, …, v_n)$  is the vertex sequence of the directed walk.
	\item A \textit{(directed) trail} is a \textit{(directed) walk} in which all edges are distinct.
	\item A \textit{(directed) simple path} is a \textit{(directed) trail} in which all vertices are distinct.
	\item A \textit{(directed) cycle} is a non-empty directed trail in which the only repeated vertices are the first and last vertices \cite{bender2010lists}.
\end{itemize}   
Consequently, from these definitions implies that a directed walk can consist of several simple paths and cycles.

In our clickstream data from on-board application, each sequence is the succession of state transitions in the underlying FSA of that application. Thus, each state transition can be considered as an edge in a FSA graph and the whole sequence can be considered as a directed walk on that graph. Finally, these directed walks can be reduced to their components, namely simple paths and cycles. We introduce a method to split the sequences into their components.

\subsection{Algorithm}
As we discussed previously in section 2, a sequence $q$ is the consecutive occurrence of the states of the FSA graph of a specific application during a drive. We assume that for each two consecutive states $u$ and $v$ with $v$ appearing after $u$ in sequence $q$, there exists a directed edge $e=(u,v)$ between those states in the FSA graph.

With algorithm \ref{alg1}, we split a sequence into simple paths and cycles. In general, other splits are also possible. The example below shows the result of our algorithm and an alternative split for a walk on an example graph depicted in Figure \ref{fig1}.
\begin{figure}
	\centering
	\includegraphics[height=3cm]{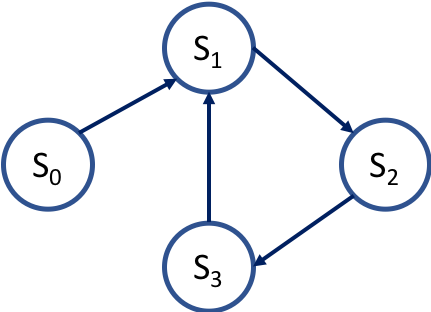}
	\caption{A small directed graph $G_1$.}
	\label{fig1}
\end{figure}

Sequence $q = (S_0,S_1,S_2,S_3,S_1,S_2,S_3,S_1,S_2,S_3,S_1,S_2)$ is a directed walk on the graph $g$ in  Figure \ref{fig1}. The algorithm \ref{alg1} reduces this sequence in following components:
\begin{itemize}
	\item a simple path: $(S_0,S_1)$
	\item three cycles: $(S_1,S_2,S_3,S_1)$
	\item a simple path: $(S_1,S_2)$
\end{itemize}
An alternative split of the same sequence can be:
\begin{itemize}
	\item a simple path: $(S_0,S_1, S_2,S_3)$
	\item two cycles: $(S_3,S_1, S_2, S_3)$
	\item a simple path: $(S_3, S_1,S_2)$
\end{itemize}

\begin{algorithm}[H]
\SetAlgoLined
	\KwData{given a sequence $q$}
	list of visited vertices $vs$ = [ ]\;
	list of paths $ps$ = [ ]\;
	list of cycles $cs$ = [ ]\;
	\For{each state $u$ in $q$}{
		\If{$u$ in $vs$ and $u$ is the first element of $vs$}
			{add $u$ to the end of $vs$\;
			add $vs$ to $cs$\; 
			set $vs$ = $[ u]$\;}
		\eIf{$u$ in $vs$ and $u$ is equal to the $i^{th}$ element of $vs$}{
			add $u$ to the end of $vs$\;
			add $(v_0,...v_{i-1})$ in $vs$ as a simple path to $ps$\;
			add $(v_{i},..., u) $ in $vs$ as a cycle to $cs$\;
			set $vs$ = $[ u]$\;
		}{
			add $u$ to $vs$\;
		}
	}
	\If{$vs$ is not empty}{add $vs$ as a simple path to $ps$}
	\KwResult{return the lists $ps$ and $cs$ as graph components of sequence $q$ }
	\
	\caption{Sequence to Graph Components}
	\label{alg1}
\end{algorithm}

\

If we put the simple paths from first splitting together, we will have a simple path from vertex $S_0$ to the vertex $S_2$, $(S_0,S_1, S_2)$. Doing the same for alternative splitting, we will have $(S_0,S_1, S_2,S_3,S_1,S_2)$ which is no longer a simple path and include a cycle.
We can show that the algorithm \ref{alg1} always produce the maximal number of cycles. For this, we start with some theorems.

\begin{theorem}
	a simple path $p_1$ is a part of another simple path $p_2$ if and only if, the edges in $p_1$ appear in the same order in $p_2$.
	\label{theo1}
\end{theorem}
\begin{proof}
	By definition, every edge in a simple path starts with the end vertex of its previous edge. Thus, for two consecutive edges $e_1$ and $e_2$ in $p_1$, we have $e_1=(u,v)$ and $e_2=(v,w)$. If the path $p_2$, has other edges between $e_1$ and $e_2$, then it contains a cycle and it is not any more a simple path.
\end{proof}

\begin{theorem}
	a cycle $c_1$ with $m$ the number of its edges can not be part of an other cycle $c_2$ with length of $n$ where $ n>m$.
	\label{theo2}
\end{theorem}
\begin{proof}
	Every vertex on a directed cycle has per definition exactly one incoming edge and one outgoing edge. If a cycle is part of a longer cycle that means, there are edges and vertices on the longer one which are not part of the shorter cycle. A cycle is a connected component. Therefore, these extra edges in the longer cycle must be also connected to the edges and vertices of the shorter cycle. This contradicts the condition of exactly one in and out edges on each vertex.
\end{proof}

From Theorem 2 follows that, a cycle is definite and clearly defined just through the set of its edges.

\begin{theorem}
	Considering only the simple paths found by Algorithm  \ref{alg1}, these paths build a connected path from start to the end of the original walk (sequence).
	\label{theo3}
\end{theorem}
\begin{proof}
	Considering a vertex $v$ on the given walk $q$, which is not the last vertex, we have only one of the three possibilities:
	\begin{enumerate}
		\item either $v$ is starting point of a cycle,
		\item or $v$ is the starting point of the next simple path,
		\item otherwise, $v$ is in the middle of a path or cycle.
	\end{enumerate}
In case 1, at the end of the cycle, we will end again at vertex $v$. So that $v$ will be again either a case 1 or 2. Considering only paths, that means a vertex at the end of a simple path, which is not the last vertex of the walk $q$, is the starting point of the next simple path. This concludes that putting all simple paths together will result in a connected path from start to the end.
\end{proof}
\begin{theorem}
	The Algorithm  \ref{alg1} identifies the maximum number of cycles on a walk $q$.
	\label{theo4}
\end{theorem}
\begin{proof}
	Assume there exists a splitting of walk $q$ with $n$ cycles more than the result of algorithm  \ref{alg1}. Without loosing the generality, we consider the n cycles with $n$ separate starting nodes in the same order in which they appear on the walk $q$. So that, we have cycles $(v_0,...v_0)$, $(v_1,...v_1)$, ..., $(v_n,...v_n)$. Suppose there is a simple path before the first cycle $(v_0,...v_0)$. The algorithm  \ref{alg1} has to find cycle $(v_0,...v_0)$ or it finds another cycle earlier. If it finds another cycle $(u,...u)$ earlier, $(u,...u)$ can not be happen before $(v_0,...v_0)$, otherwise we have $n+1$ cycles in the walk $q$. From Theorem 2 it can not be part of $(v_0,...v_0)$. This means that, either it has its starting point on the path and ends on the cycle, $u,...(v_0,...u, ...v_0)$ or it is completely inside the cycle $(v_0,...u,...,u, ...v_0)$. In both cases, the algorithm will find cycle $(u,...u)$, instead of $(v_0,...v_0)$. That means the cycle $(u,... u)$ appears before other $n-1$ cycles. The algorithm either has to find the rest or we can induce our argument to the next cycle.
\end{proof}

It is to mention that splitting the walk differs from splitting the underlying graph into its components. The example in Figure \ref{fig2} illustrates the difference. Let's compare two walks $w_1$ and $w_2$ on the graph $G_2$ depicted in Figure-\ref{sub1}. Both walks start at vertex $S_0$ and enter the cycle at two different vertices $S_1$ and $S_3$. The Figure-\ref{sub2} illustrates the two walks with the edges of each walk marked with different colors. In Figure-\ref{sub3}, we see the splitting of the walks as a result of algorithm \ref{alg1}. Both walks are split into three components. Two simple paths and a cycle. They share the cycle but have different paths. The ending path of the walk $w_2$ is the part of ending path of the walk $w_1$. Figures~\ref{sub4} and~\ref{sub5} shows the components of walks $w_1$ and $w_2$ from~algorithm \ref{alg1}. On the other hand, we can consider the graph of each walk and split it by the cycles and non cycle parts of the graph as depicted in Figure-\ref{sub6} and Figure-\ref{sub7}. In this case, the two walks share two similar components and only differ by their starting paths. It is just a matter of perspective, what we are interested in and what we see as more similar. In the splitting in Figure-\ref{sub6} and Figure-\ref{sub7}, some information from original walks are lost. Both walks are running through exactly the same vertices, but their length varies. This means that some vertices are visited more often in one walk than the other one, which is visible in Figure-\ref{sub2}, as we are walking part of the cycle for a second time. This information is coded in the split in Figure-\ref{sub3}, so that the last paths differ. In this paper we consider that walk $w_1$ and walk $w_2$ only share one cycle, which they enter at different vertices.

We use algorithm \ref{alg1} to split sequences into their components. In the next section, we show how we use a native graph database to save the sequences by their components and perform sequence analysis by querying the graph database.

\begin{figure}[t]
	\centering
	 \begin{subfigure}[b]{0.35\textwidth}
		\includegraphics[width=\textwidth]{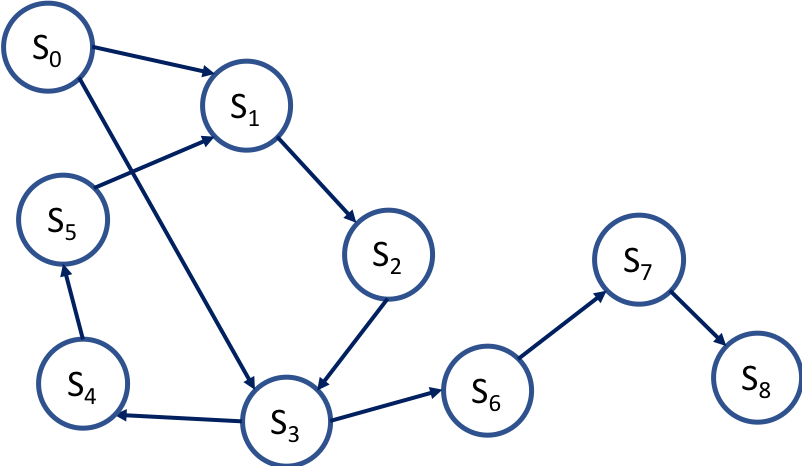}
		\caption{An Example of a graph $G_2$}
		\label{sub1}
	\end{subfigure}
	\linebreak
	\linebreak
	\begin{subfigure}[b]{0.35\textwidth}
		\includegraphics[width=\textwidth]{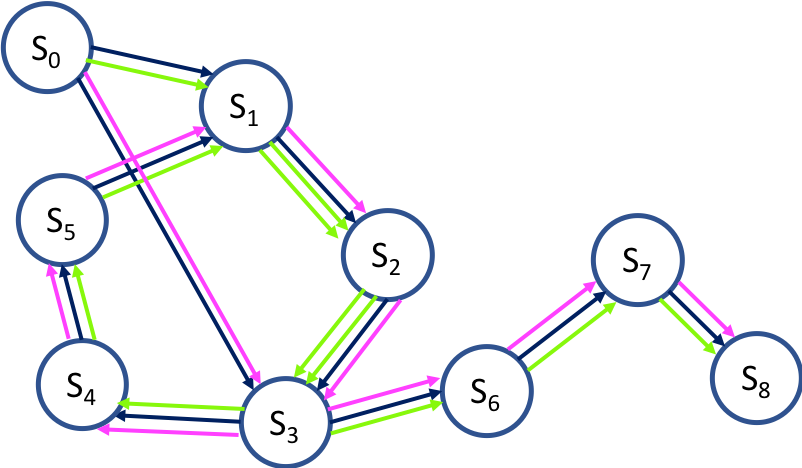}
		\caption{Two possible walks on $G_2$ (green and magenta).}
		\label{sub2}
	\end{subfigure}
	~~~~~~~
	\begin{subfigure}[b]{0.44\textwidth}
		\includegraphics[width=\textwidth]{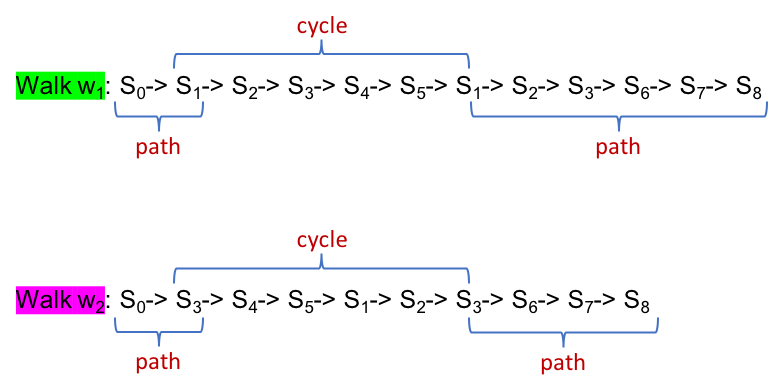}
		\caption{Two possible walks depicted in (b) and their components.}
		\label{sub3}
	\end{subfigure}
	\linebreak
	\linebreak
		\begin{subfigure}[b]{0.45\textwidth}
		\includegraphics[width=\textwidth]{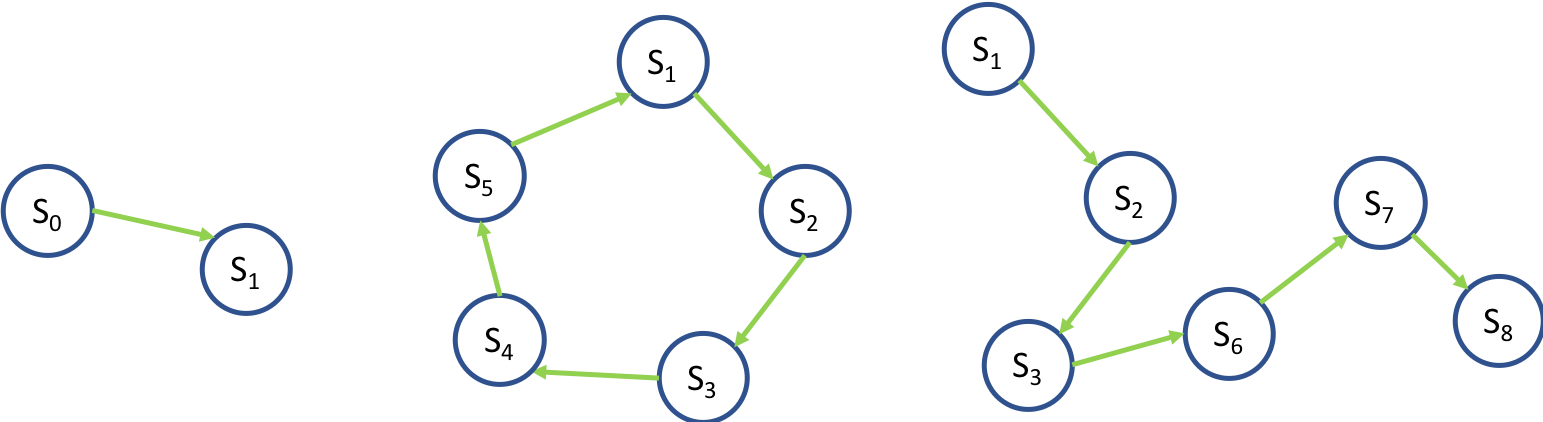}
		\caption{Components for walk $w_1$ from algorithm~ \ref{alg1}.}
		\label{sub4}
	\end{subfigure}
	~~~~~~~
	\begin{subfigure}[b]{0.45\textwidth}
		\includegraphics[width=\textwidth]{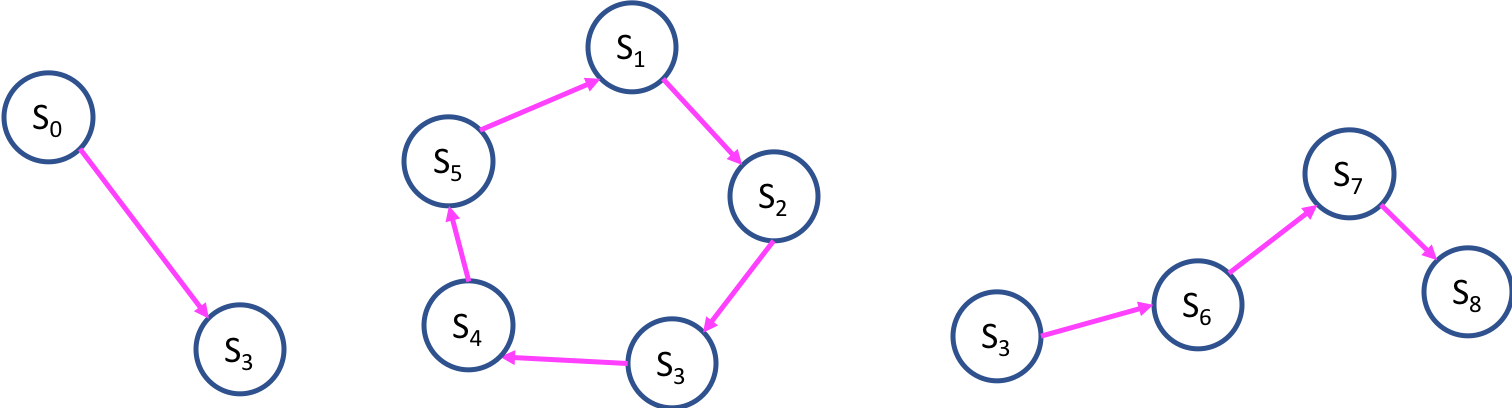}
		\caption{Components for walk $w_2$ from algorithm~ \ref{alg1}.}
		\label{sub5}
	\end{subfigure}
	\linebreak
	\linebreak
	\begin{subfigure}[b]{0.45\textwidth}
		\includegraphics[width=\textwidth]{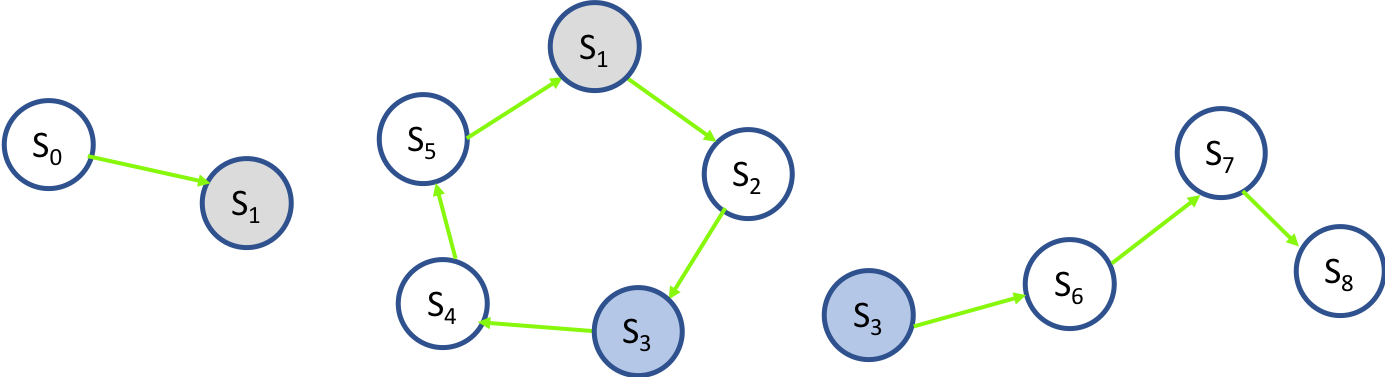}
		\caption{Graph components of walk $w_1$}
		\label{sub6}
	\end{subfigure}
	~~~~~~~
	\begin{subfigure}[b]{0.45\textwidth}
		\includegraphics[width=\textwidth]{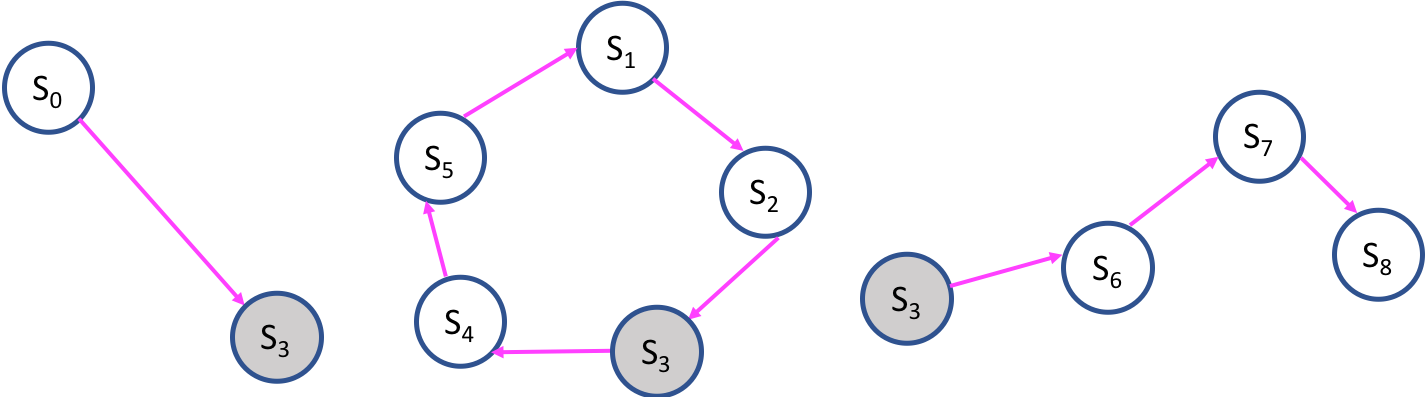}
		\caption{Graph components of walk $w_2$}
		\label{sub7}
	\end{subfigure}

	\caption{Graph components versus walk components from algorithm \ref{alg1}.}
	\label{fig2}
\end{figure}

\section{Modeling sequences in graph database}
In native graph databases, graph models play a central role. The graph model represents how we look into the data and which questions and queries we are interested in. In other words, our problem domain is reflected in our graph model~\cite{robinson2013graph}. In this chapter, we discuss several modeling scenarios for our sequence components and compare them with each other. we start with the problem description and discuss other aspects like performance and flexibility in the following sections.

\subsection{Problem description}\label{problem}
In our clickstream data, we are interested in the ways customers use our application in the car. The key analytical questions are, e.g., common problems like long common paths, or cycles, to fulfill a task, or clustering the customer according to their usage pattern of the application.

In the literature, there are methods and topics which address similar or related problems, such as funnel analysis. In contrast to our method, in funnel analysis, we first select a specific path and then analyze the behavior of the customers on that path~\cite{mah2006method}.
 
The majority of the works on customer behavior and clickstream analysis is either focused on sequence analysis by sequence vectorization~\cite{montgomery2004modeling,lu2005mining,ting2005ubb,wang2013you}, or comparison of some key performance indicators~\cite{moe2002bayesian,heer2002separating}. In most of these studies the patterns inside a sequence are coded and compared with each other but the meaning and shape of a single sub-pattern in the sequence is lost or not considered. In the method we introduced in the section \ref{comps}, we split sequences into their underlying components such as simple path and cycle. These components can be compared, classified and grouped in the process of sequence analysis. For this, we save the sequences in form of their components in a graph database. Our models have a hierarchical structure and at each level of hierarchy we can save the corresponding meta data for that level of hierarchy. Figure \ref{fig7} shows a modeled hierarchy for the clickstream data from vehicle's onboard applications. At the top of hierarchy we have customers. Each customer performs several drives. During each drive different applications are used. For each application usage we have the corresponding clickstream sequence which can be split into its underlying simple paths and cycles by algorithm \ref{alg1}. In the example from Figure \ref{fig7}, application 1 in drives 1 and 2 shares components 2 and 3. Same application in drive 3 has totally different components than in drive 1 and 2.

At each level of hierarchy we can have different attributes. A customer can have gender, age and address. A drive has a start and an end time, and start and destination positions. Application sequence has the attributes of application category, e.g. \emph{Navigation}, \emph{Entertainment}, etc. Each application can have its own specific attributes. \emph{Navigation}, for example, can have the binary status if a routing calculation during the drive has started or not.
\begin{figure}[h]
	\centering
	\includegraphics[width=0.7\textwidth]{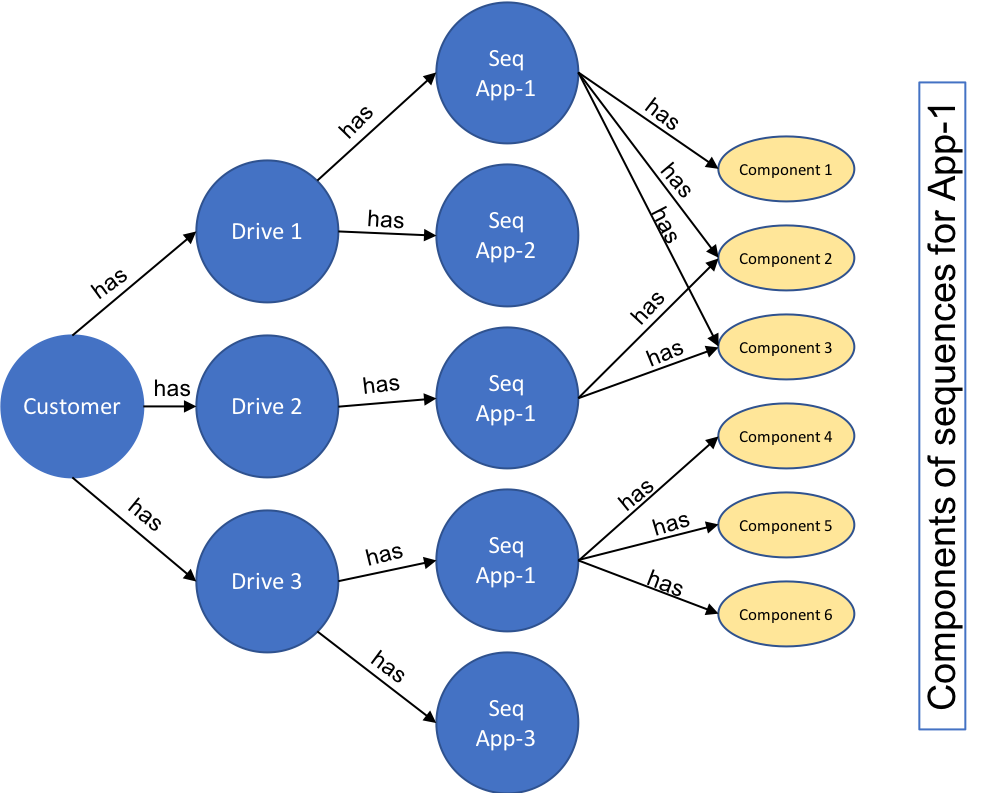}
	\caption{A modeled hierarchy for clickstream data.}
	\label{fig7}
\end{figure}

Each sequence itself consists of a finite number of components such as simple paths and cycles from the walk on the FSA of the corresponding application. In the next section we describe several scenarios how we can model these components in a graph database.

\subsection{Modeling the sequence components in a graph database}\label{models}

As we discussed in chapter 3, a simple path $p$ can be clearly defined by the ordered list of its vertices, where the order defines visiting time of the vertex. A simple path has clear start and end point. From  Theorem~\ref{theo2} follows that a cycle can be defined clearly by a set of its edges, while the order of edges in a set is not important. A cycle does not have a defined start and end point. It can be entered at any vertex. We start our models according to these mathematical definitions.

Suppose we have a set of sequences $Q$ to analyze, with $|Q|$ as the number of sequences in $Q$. The Algorithm \ref{alg1} results in $|C|$ distinct cycles and $|P|$ distinct simple paths after computing on set $Q$. Lets be $c_e$ and $c_v$ the number of all edges and vertices of the cycles in $C$ (with duplicates). Analogously, we have $p_e$ and $p_v$ as the number of all edges and vertices of the paths in $P$ respectively (with duplicates). The sequences in $Q$ are walks on a FSA graph. The FSA hast $|V|$ vertices and $|E|$ edges. In the following we use these numbers to compare different modeling variants.

\subsubsection{Model variant 1}
In the first variant we start modeling directly on FSA graph. We model distinct states of an FSA as defined vertices in graph database. Each path and cycle is marked with a new different edge between the vertices of FSA. Several edges can exist between two vertices. A simple example with three components in Figure \ref{fig8}, one cycle and two paths, demonstrate our model in variant 1. As we can see between vertex $S_2$ and $S_3$ we have three edges with different types for each component. In this model we have to create $|V|$ vertices and $p_e + c_e$ edges in the graph database. 
\begin{figure}[h]
	\centering
	\includegraphics[width=0.7\textwidth]{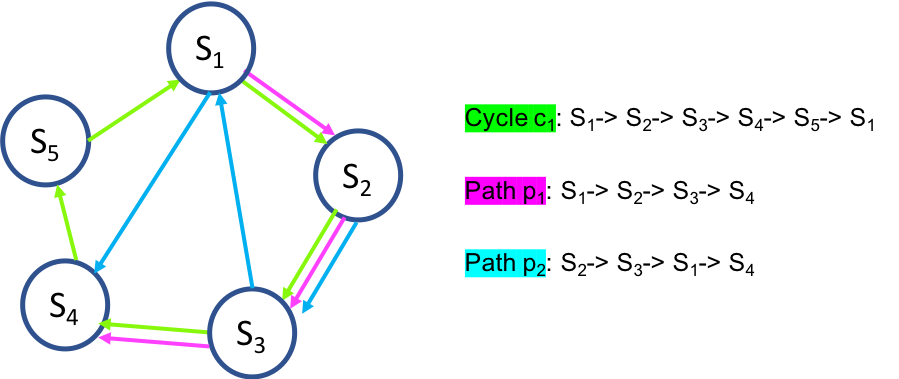}
	\caption{Cycles and paths in model variant 1.}
	\label{fig8}
\end{figure}
\subsubsection{Model variant 2} 
In the second variant, we save each component as separate vertices and edges as depicted in Figure \ref{fig10}. As a result, we have to create $p_v + c_v$ vertices and $p_e + c_e$ edges in the graph database.
\begin{figure}[h]
	\centering
	\includegraphics[width=0.7\textwidth]{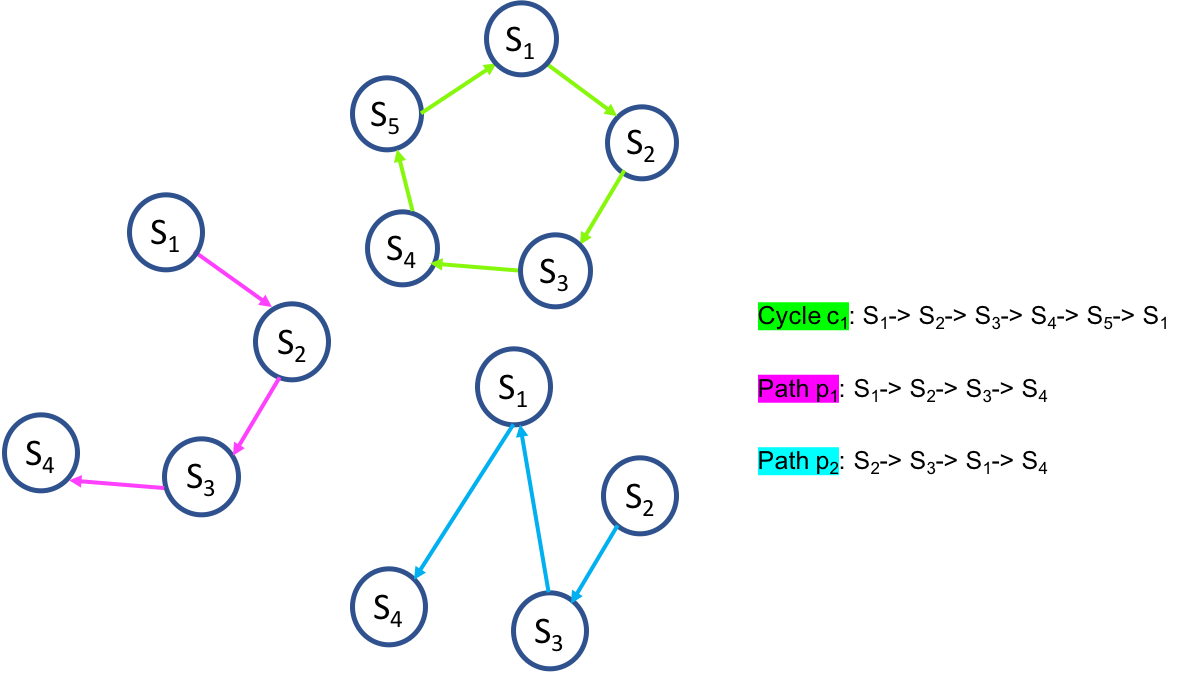}
	\caption{Cycles and paths in model variant 2.}
	\label{fig10}
\end{figure}
\subsubsection{Model variant 3} 
In the third model variant, instead of the vertices of the FSA, we save the distinct edges of the FSA as nodes (transit-nodes) in a graph database. Each path and cycle is represented as a node and is related to its corresponding transit-nodes. That means, components can share transit-nodes.
In Figure \ref{fig9}, the same example from Figure \ref{fig8} is modeled as variant 3. For an FSA with $|V|$ vertices as directed graph, the maximum number of edges $|E|$ is limited to 2-permutation of $|V|$ and therefore can be maximally $|V|*(|V|-1)$. The number of edges stays the same as in version 1, $p_e + c_e$. We have also to add extra vertices for representation of each path and cycle, which means $|C|+|P|$ extra nodes in database.
\begin{figure}[h]
	\centering
	\includegraphics[width=0.5\textwidth]{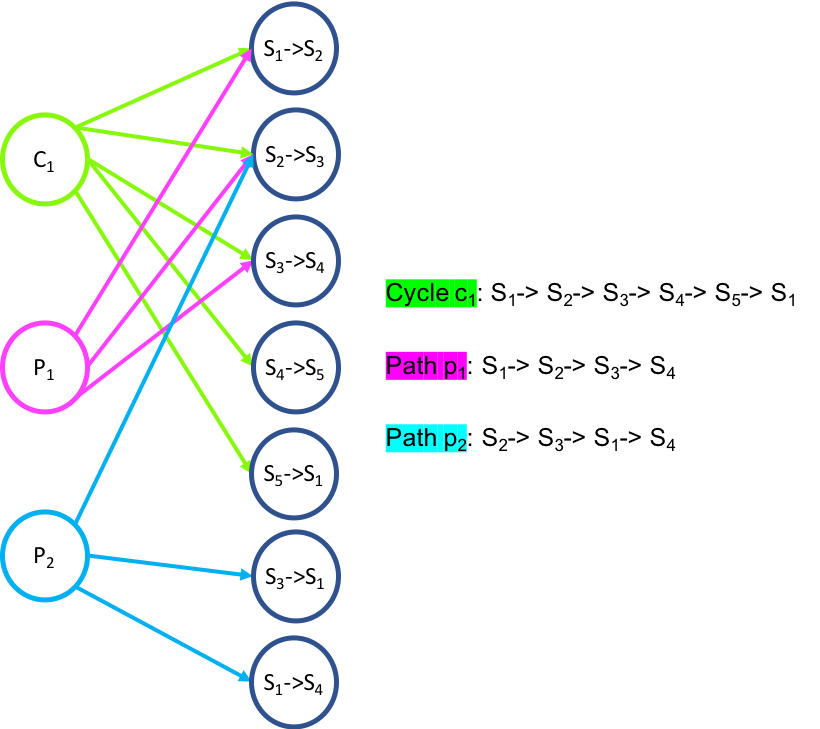}
	\caption{Cycles and paths in model variant 3.}
	\label{fig9}
\end{figure}

\subsection{Implementation of models in a graph database}
To prototype our models, we used the Neo4j community edition~\cite{neo4j}. In general, the concept is independent from the graph database platform but still variation in the implementation of graph databases can influence the performance on different models. To clarify the used vocabulary in the discussion below we go through a short introduction of some conventions in Neo4j. In Neo4j, there are nodes and relationships, which can be considered as equivalent to vertices and edges in a graph.
Each node or relationship has a type which acts like a classification, that means we can create several nodes and relationships of the same type. Each node and relationship can also have several attributes of different data types  such as name (\textit{String}), creation time (\textit{DateTime}), visited (\textit{Boolean}) etc. We use \emph{Cypher graph database query language}, widely used in Neo4j community~\cite{robinson2013graph}, for our expressions in this context. The below Cypher expression creates two nodes and a relationship between them. The nodes are of type \textit{State} and the relationship is of type \textit{Cycle}. $a$ and $r$ are variables. Type \textit{State} has an attribute \textit{Name}, which for the first node is \textit{Address Book} and for the second node is \textit{Search Field}. The relationship with the type \textit{Cycle} has an attribute \textit{ID} with value $101$.

\begin{verbatim}
CREATE (a:State {Name:'Address Book'})-[r:Cycle {ID:101}]->
(:State {Name:'Search Field'}) 
\end{verbatim}

\subsubsection{Injection of cycles and paths}
Suppose we have a set of $|C|$ distinct cycles and ${|P|}$ distinct paths to add to the graph database. 
\subparagraph{Injection in model variant 1:}
In model version 1, we can create the $|V|$ states of the FSA graph at once. Then, for each cycle or path, it is enough to \textit{select} its vertices in database and \textit{create} the edges. Due to the relative small number of states, this task is performed in constant time. But Neo4j is not designed for multiple-edges modeling. For creating multiple relationships between two nodes, there are several possibilities:
\begin{itemize}
	\item multiple relationships with same type and different e.g. name attributes,
	Neo4j does not support creation of such multiple relationships between two nodes from the available drivers. It is possible to create such relationships directly from the Neo4j desktop.
	\item multiple relationships with different types.
	\item single relationship with a list of attributes. Having paths as attribute entries in a list makes the queries on searching and selecting those paths inefficient because the graph database is optimized for queries directly performed on nodes and relationships.
\end{itemize}	
	In both cases, with multiple relationships,for both single type or multiple types, we have to check if the new path or cycle already exists. As we know from Theorem~\ref{theo1}, a shorter path can be part of a longer path. Thus, such paths will cause problems by selecting them in model variant 1, because all the longer ones will be selected along with it. At the end, we have to perform more complicated queries to find out which of them are longer ones. Another disadvantage of this model is the assignment to its higher level. Because the distinction between components is modeled as edges, it is not possible to assign the parent level nodes (drives / sequences) by a relationship. We have to put this information at the attribute level, which again withdraws the benefits of a graph structure.

\subparagraph{Injection in model variant 2:}
In the model version 2, for every distinct component we create nodes and relationships between those nodes as a new component. The path or cycle can be distinguished by, e.g., its ID or Name at the attribute level on nodes and relationships. By every injection we have to check if that path or cycle already exists in the database. It can be performed by the following expressions easily.
\begin{itemize}
	\item  For a path with n nodes:
\end{itemize}
\begin{verbatim}
MATCH p =(s:State {name: 'State Name 1', title: 'start'})-[r]->
(State {name: 'State Name i'})--> ... // all the nodes is between
(e:State {name: 'State Name n', title: 'end'})
RETURN r.name
\end{verbatim}
The paths is clearly identified through its start, end and the nodes in between.
\begin{itemize}
	\item  For a cycle with n nodes:
\end{itemize}
\begin{verbatim}
MATCH p =(:State {name: 'State Name 1')-[r]->
(State {name: 'State Name i'})--> ... // all the nodes is between
(State (:State {name: 'State Name 1'))
RETURN r.name
\end{verbatim}
The cycles need no start or end attribute, the only difference is that the start and end node of the match has to be the same. In this model we are saving the state nodes duplicated compared to version 1. On the other hand, we can assign the components to the parent nodes with an explicit relationship. Each drive can enter the cycle component at different nodes. Furthermore, the relationship between drives and their components can have timestamps to capture the order of the components in the walk. If the component dose not exist in the database, it can be added by replacing the MATCH with a CREATE.

\subparagraph{Injection in model variant 3:}
In the third variant we create the edges of the underlying FSA graph as nodes and assign them to the cycles and paths. The MATCH expressions to check if the components already exists, can have the following forms:
\begin{itemize}
	\item  For a path with n nodes.
\end{itemize}
\begin{verbatim}
MATCH (c:Paths)-[d:P]->(:T)
WHERE ALL( x IN {start_end}  WHERE (c)-[:P]-> 
		//start_end is a list of tuples for each edge 
		//with names of its start and end states.
(:T {start: x[0], end: x[1]}) )
WITH c, count(d) as pathLength
WHERE pathLength = length({start_end})
RETURN c.name
\end{verbatim}
\begin{itemize}
	\item  For a cycle with n nodes.
\end{itemize}
\begin{verbatim}
MATCH (c:Circles)-[d:C]->(:T)
WHERE ALL( x IN {start_end}  WHERE (c)-[:C]->
		//start_end is a list of tuples for each edge 
		//with names of its start and end states.
(:T {start: x[0], end: x[1]}) )
RETURN c.name
\end{verbatim}
In this model, the edges are shared, hence even if the component itself has to be created, it must be guaranteed that the node for edges are added once. For this, we use MERGE instead of CREATE. MERGE is equivalent to MATCH before CREATE. Accordingly, a CREATE operation is faster than MERGE.

\subparagraph{Injection of large amount of data:} If we have to deal with a large amount of data at once, it is better to avoid the single inserts due to lock latencies. For this, in Neo4j we can benefit from table-wise insertions, in which the node and edge information are provided row-wise in a table. The table can be uploaded in one step and will be internally processed row by row. In the case of variant 1, the problems already discussed in the case of single injection also persist in the table-wise insertion. In the model variant 2, we create for each edge of the components, a row in a table with a start, an end and a component id. In row wise insertion, first the start and the end nodes will be merged. Finally, the relationship between the start and the end of the edge will be created. 
In variant 3, we can have the edges per component as rows of a table and merge them to the database, the same way as in the single insertion.

\subsubsection{Query the graph}
Actually the main purpose of the data transformation into graph database, is to open the possibilities for data queries which in relational databases are either inefficient or even impossible. We concentrate ourselves on the example of several queries which play an important role in the analysis of customer behavior. At the same time, these queries do not perform well in traditional relational databases, particularly in the case of large amounts of data. Of course, we can always add explicit indexing on relational databases to mimic the benefits of non-relational databases. The scope of discussion here is the benefits of non-relational databases such as graph databases, and not the way we want to realize such designs.

\subparagraph{Example 1:}
Given a sequence of consequent states such $S_5 -> S_3 -> S_7$ as a simple path, we are interested in the customers who traversed exactly this simple path during the application usage. In a relational database, the data as described in section~\ref{data}, has to be sorted per session and timestamp. Even if, the data is stored sorted, we still have to search the whole table (e.g. several T bytes) for the start state of the simple path and check the consequent states afterwards. In a graph database such simple paths can be searched at once with all nodes and edges. Experiments from our prototype show that the number of graph components converts. That means the search space is much smaller than the original table in the relational database.

\subparagraph{Example 2:}
We are interested to find out how many different paths the customers traverse between two states of interest $S_i$ and $S_j$. The query itself will be similar to the example before, but the variation depends strongly on whether we consider the loops and the multiple traversal of the loops or not. Identifying loops is not straightforward. However, the component based storage of the clickstream data makes it possible to identify such structures.

\subparagraph{Example 3:}
We are interested to find out in which loops the customers are stuck most of the time. Firstly, finding the loops itself is not straight forward. Secondly, we have to be able to count the loops as the same loop even if, they appear with different starting points in the sequence. Here is an example; $S_5 -> S_3 -> S_7 -> S_5$ and $S_7 -> S_5 -> S_3 -> S_7$ are the same cycle. For these kinds of analysis we have to extract data and solve our questions algorithmically. The components wise storage of the sequences allows us to perform such analysis as queries. We can also combine it with similar reasoning as funnel analysis, such as considering loops appear in a path of interest.

\subparagraph{Example 4:}
One of the main focus of click stream analysis and customer behavior analytics in literature is sequence pattern mining. They usually use methods of vectorization such as k-grams to map the sequences to vectors and define metrics on them~\cite{wang2017clickstream}. Finally they use known vector based pattern mining methods such as clustering. In this work, we have introduced component wise analysis of the sequences. The main hypothesis is that the sequences consist of cycles and simple paths which they share. In other words, we can compare and group the sequences according to their common components. In our prototype, we can cluster the sequences into clusters of the exact same components. The clustering itself can be performed by a single query and can be restricted easily for further analysis.

The first two examples are query based analysis which show the benefits of a graph database storage while the third and forth example highlight the benefit of component-wise analysis of the sequences.

\section{Result}
For our prototype we used a subset of the data described in section~\ref{data}. The subset consists about 200~k drives within two weeks and is about $5$ GB in size. For the storage in the graph database, the following steps were performed:
\begin{itemize}
	\item Extracting the sequences per drive.
	\item Deriving the components of the sequences.
	\item Insertion of the sequences and their components into the graph.
	\item Insertion of customer-drive nodes into the graph and assignment to the sequences.	
\end{itemize}
The created result in the graph database is less than 500 MB large. We have achieved roughly a space reduction of 10x. We implemented all the three model variants described in section~\ref{models}. All three models result in a similar storage amount of $<$ 500 MB. In what follows, we compare the model variant 2 and 3 by performing the queries from examples 1 to 4. Table~\ref{tab:problem} gives an overview of the subset used for our experiments according to the numbers introduced in section~\ref{problem}. Table~\ref{tab:models} summarizes the number of elements in each model. It also shows if the element is modeled as a relationship (edge) or a node.
\begin{table}[h]
	\centering	
	\footnotesize 
	\begin{tabular}{|c|c|}
		\hline
		Quantities& number		\\
		\hline
		$|Q|$ number of sequences $\equiv$ drives	&224265\\
		\hline
		$|C|$ number of distinct cycles	&3767\\
		\hline
		$|P|$ number of simple paths	&7438\\
		\hline
		$c_e = c_v$ total number of edges or vertices in cycles	&16077\\
		\hline
		$p_e = p_v -1$ total number of edges or vertices in paths	&25736\\
		\hline
		$V$ number of vertices in FSA $\equiv$ states	&124\\
		\hline
		$E$ number of edges in FSA $\equiv$ transits	&1691\\
		\hline
	\end{tabular}
	\caption{Overview of the statistics for the prototyped subset.}
	\label{tab:problem}	
\end{table}

\begin{table}[h]
\centering	
\footnotesize 
\begin{tabular}{|c|c|c|c|c|c|c|c|c|}
\hline
\multicolumn{5}{|c}{}&
\multicolumn{2}{|c|}{Seq. Components}&
\multicolumn{2}{c|}{}\\
\hline
					&Vehicle		&Drove	&Drive			&Has			&Circles	&Paths	&State	&Transit\\
\hline
NodesV2	&85975		&-				&224265			&-				&-			&-			&49375&-\\
\hline
EdgesV2	&-					&224265		&-					&1570774	&16077	&25736	&-&-\\
\hline
NodesV3	&85975		&-				&224265			&-				&3767	&7438		&-&1691\\
\hline
EdgesV3	&-				&	224265		&-					&1570774	&16077	&25736	&-&-\\
\hline
\end{tabular}
\caption{Model Variant 2 compared with Model Variant 3}
\label{tab:models}	
\end{table}

In the first query from example 1, we search the database for all drives which have passed through the simple path $S_{1710} -> S_{552} -> S_{574}$, where the numbers 1710, 552 and 574 are the states' IDs. The search query in model variant 2 looks like this:
\begin{verbatim}
MATCH p=(v:Vehicle)-[:Drove]->(:Drive)-[:Has]->(:State)-[*0..]->
(:State {stateId:1710})-->(:State {stateId:552})-->
(:State {stateId:574})-[*0..]->(:State)
RETURN p
\end{verbatim}
The same search in model variant 3 looks like this:
\begin{verbatim}
Match (v:Vehicle)-[:Drove]->(:Drive)-[:Has]->
(c)-->(:T {name: "1710_552"}),
(c)-->(:T {name: "552_574"})
with (v:Vehicle)-[:Drove]->(:Drive)-[:Has]->(c)-->(:T) as p
return distinct p
\end{verbatim}
In both models it takes few milliseconds to query the data. Figure~\ref{Q1} shows the results of query 1. The query finds 4 paths and 6 drives.

\begin{figure}[t]
	\centering
	\begin{subfigure}[b]{0.8\textwidth}
		\includegraphics[width=\textwidth]{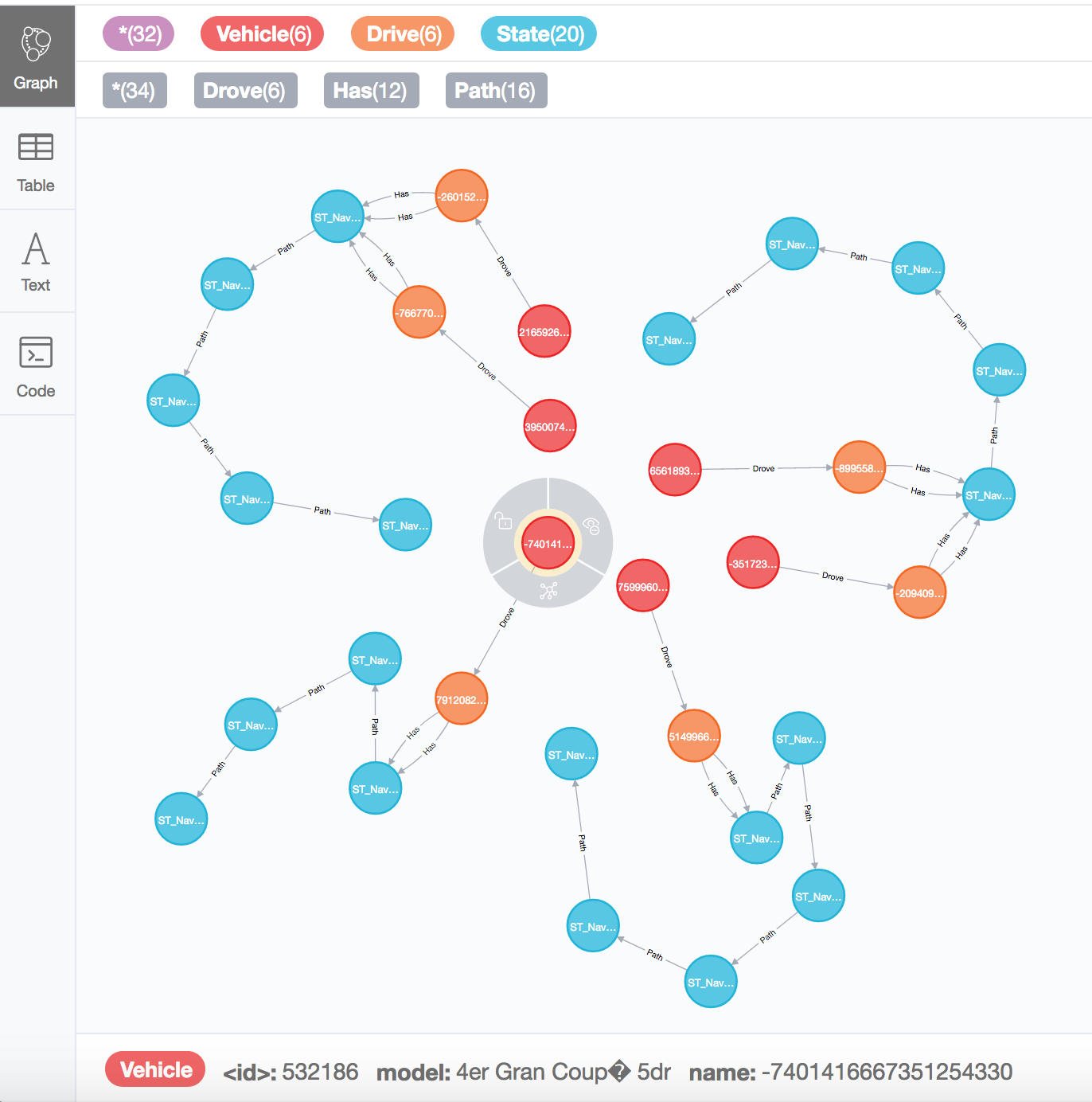}
		\caption{Result of Query 1 in model variant 2.}
		\label{Q1_V2}
	\end{subfigure}
	\linebreak
	\linebreak
	\begin{subfigure}[b]{0.8\textwidth}
		\includegraphics[width=\textwidth]{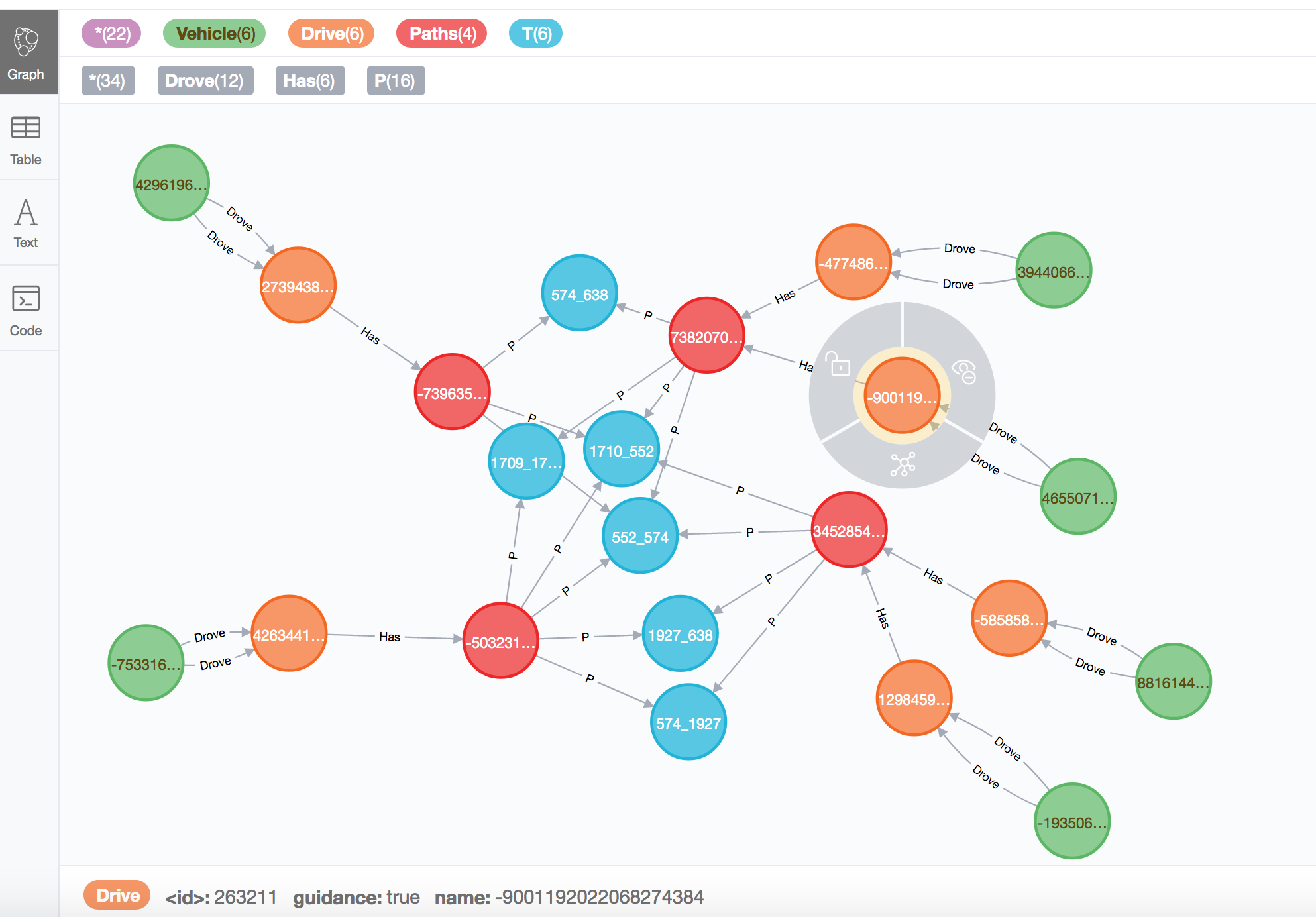}
		\caption{Result of Query 1 in model variant 3.}
		\label{Q1_V3}
	\end{subfigure}

	\caption{Result of query 1 in different models. Four paths is found. Six drives and vehicles has these paths as their components.}
	\label{Q1}
\end{figure}

In the second query from example 2, we search the database for all paths going through state $S_{1710}$ and state $S_{574}$. In Figure~\ref{Q2}, we compare the result of query 2 in both model variants. In model variant 2, the query finds 8 simple paths and a cycle. In model variant 3, we have 9 paths and a cycle. The reason for the difference is, that in model variant 3, it is difficult to consider the order of the transit nodes. For this, it finds a path in which the state $S_{574}$ is visited before the state $S_{1710}$.
Below is the possible query for variant 2:
\begin{verbatim}
MATCH (:State {stateId:1710})-[c*]->(:State {stateId:574})
with (:Vehicle)-->(:Drive)-->(:State)
-[*{compHash: head(c).compHash}]->() as p
RETURN  p
\end{verbatim}
The same search in model variant 3 looks like this:
\begin{verbatim}
Match (v:Vehicle)-[:Drove]->(:Drive)-[:Has]->
(c)-->(:T {start: "ST_Nav_DestInput_FTS_Result"})
,(c)-->(:T {end: "ST_Nav_LastDestinations"})
with (v:Vehicle)-[:Drove]->(:Drive)-[:Has]->(c)-->(:T) as p
return distinct p
\end{verbatim}

\begin{figure}[t]
	\centering
	\begin{subfigure}[b]{0.70\textwidth}
		\includegraphics[width=\textwidth]{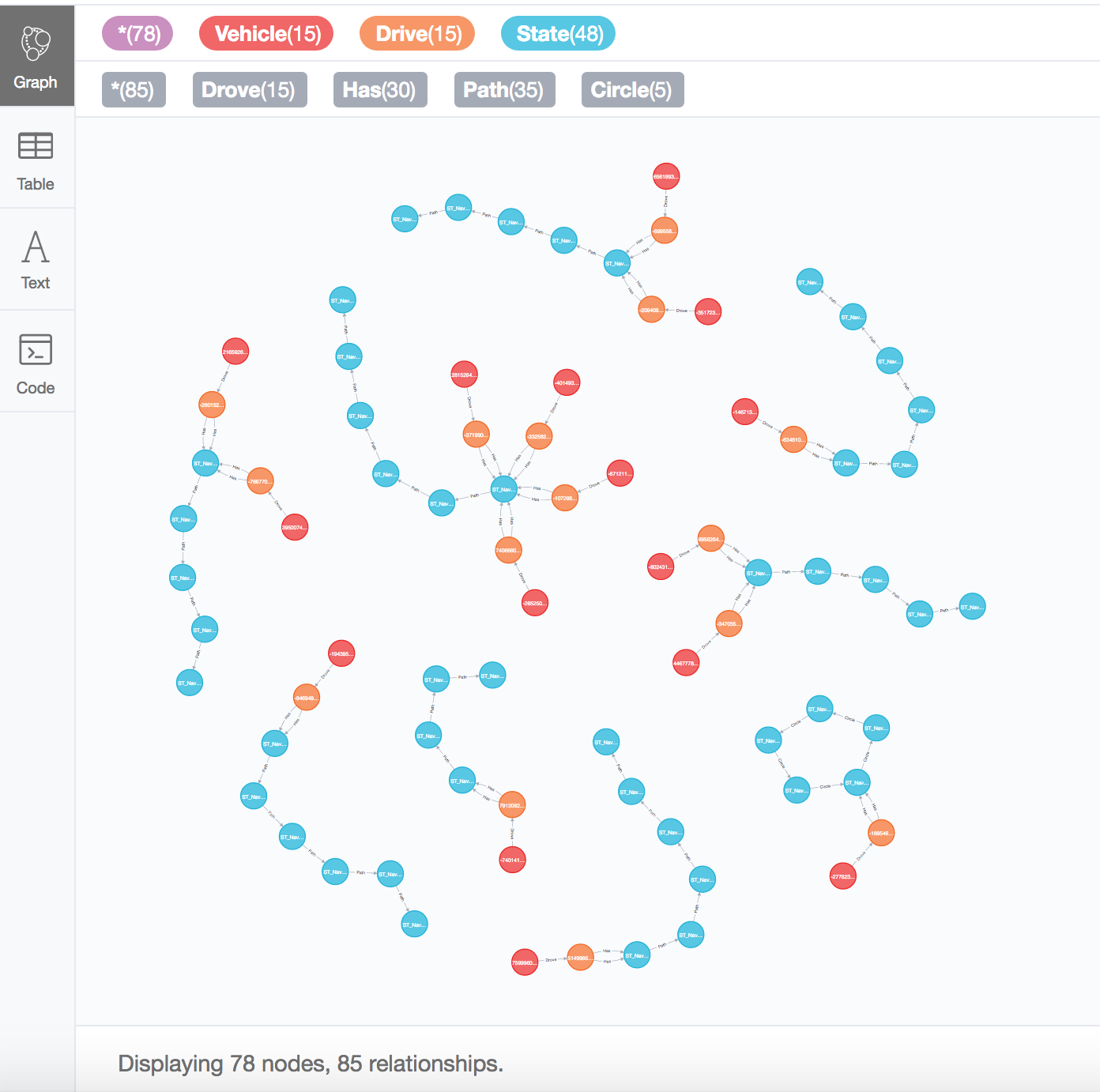}
		\caption{result of Query 2 on model variant 2.}
		\label{Q2_V2}
	\end{subfigure}
	\linebreak
	\linebreak
	\begin{subfigure}[b]{0.70\textwidth}
		\includegraphics[width=\textwidth]{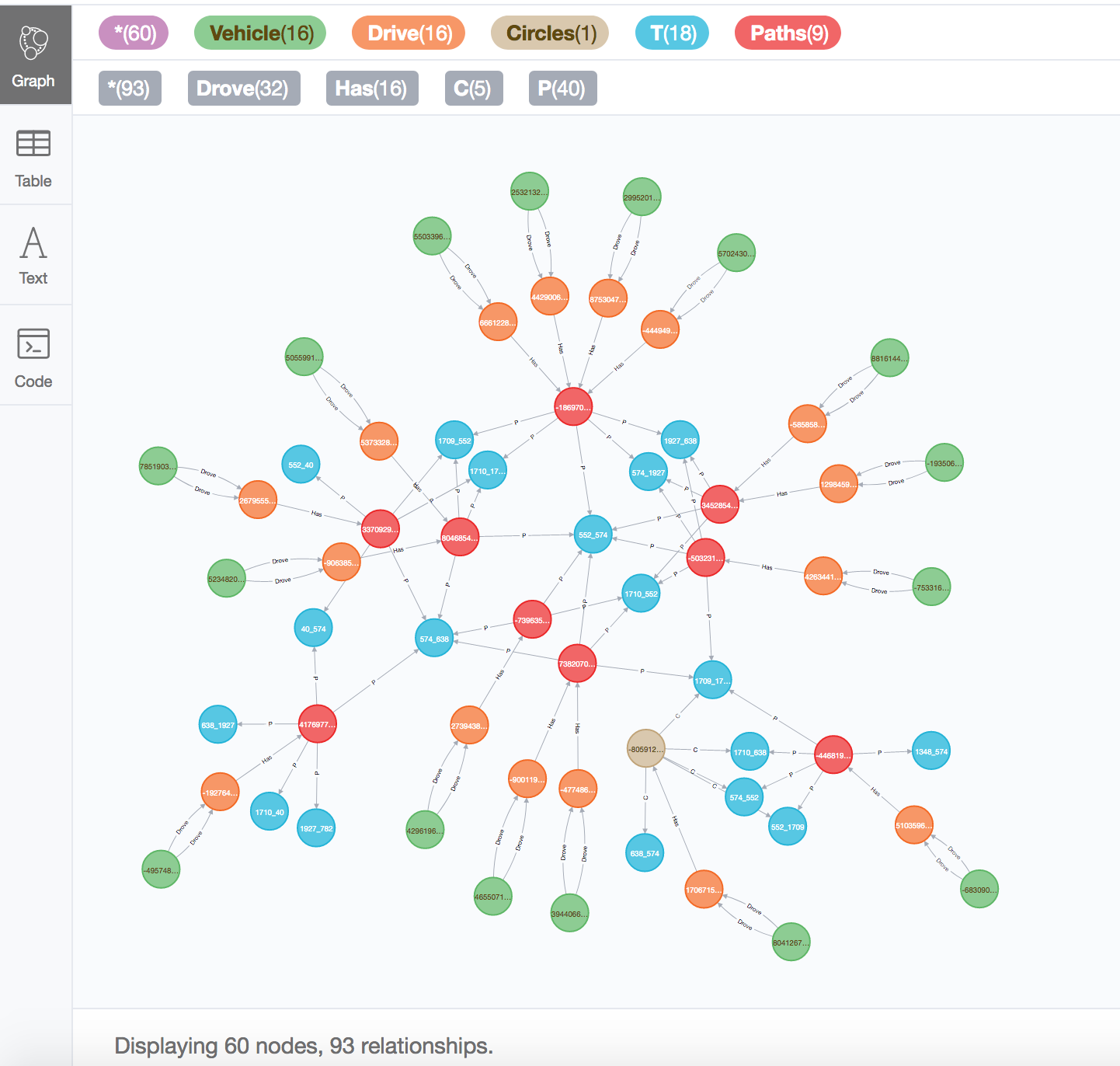}
		\caption{result of Query 2 on model variant 3.}
		\label{Q2_V3}
	\end{subfigure}
	
	\caption{Model variant 3 finds one more path, because it does not taking into account the order of visits.}
	\label{Q2}
\end{figure}

In the third query, we are looking for non trivial cycles, which are visited more than one time by several users. Figure~\ref{Q3} shows one of the cycles with six nodes from the result list with 16 users. It shows that these users have difficulty to use a specific functionality, as they are repeating these particular steps several times.
\begin{figure}[t]
	\centering
	\begin{subfigure}[b]{0.85\textwidth}
		\includegraphics[width=\textwidth]{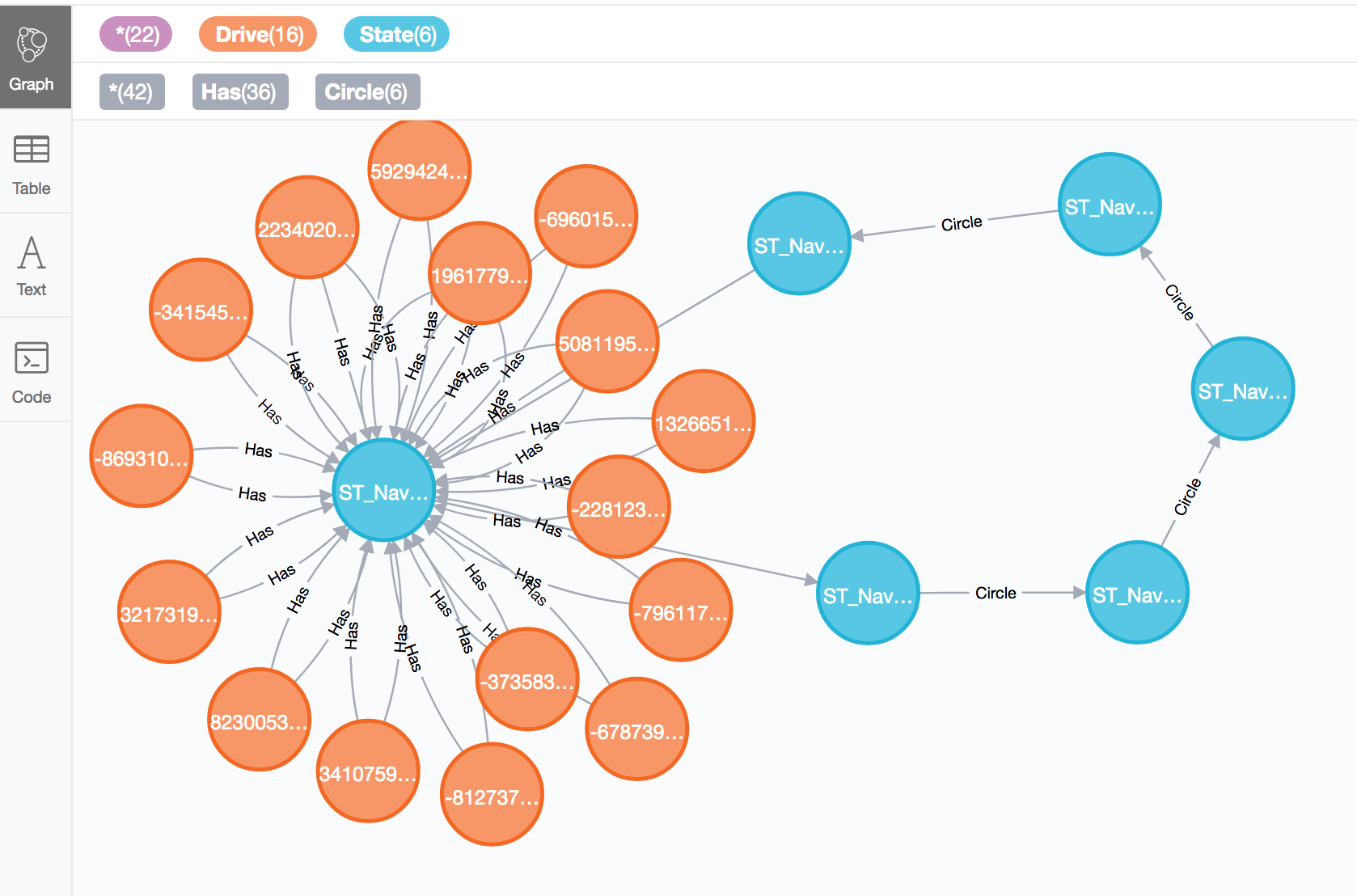}
		\caption{Result of Query 3 in model variant 2.}
		\label{Q3_V2}
	\end{subfigure}
	\linebreak
	\linebreak
	\begin{subfigure}[b]{0.85\textwidth}
		\includegraphics[width=\textwidth]{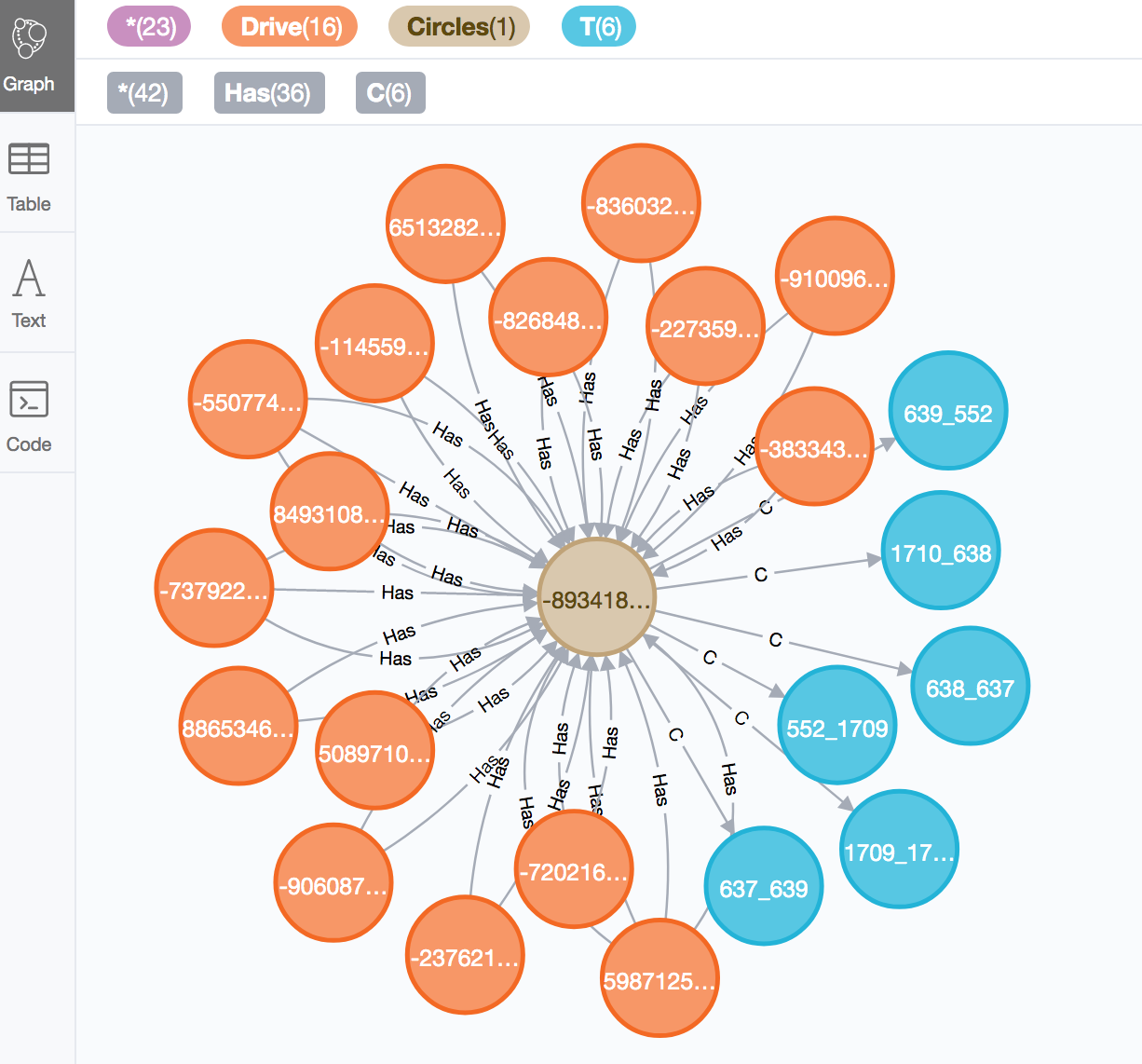}
		\caption{Result of Query 3 on model variant 3.}
		\label{Q3_V3}
	\end{subfigure}
	
	\caption{Result of query 3 on different models. A cycle with 6 nodes and 16 users who go through the cycle several times.}
	\label{Q3}
\end{figure}
Here is the Query 3 in Model variant 2:
\begin{verbatim}
Match p =(d:Drive)-[h:Has]->(s:State)-[:Circle*]->(:State) 
with id(d) as di , count(h) as NrVisit,  s.compHash as CircleName,
length(max(p))-1 as CircleLen
where NrVisit/CircleLen > 1
with CircleName, size(collect(di)) as NrDirves, 
avg(NrVisit/CircleLen) as NrVisits, CircleLen
where NrDirves > 10
Return CircleName, NrDirves, NrVisits, CircleLen
order by CircleLen desc
limit 20
\end{verbatim}
And here is the Query 3 in Model variant 3:
\begin{verbatim}
MATCH (n:Drive)-[h:Has]->(c:Circles)
with n, count(h) as ch, c
where ch > 1
with c, size(collect(n.name)) as NrDirves, 
avg(ch) as NrVisits, 
size((:Circles {name: c.name})-->(:T)) as CircleLen
where NrDirves > 10 
return c.name as CircleName, NrDirves, NrVisits, CircleLen
order by CircleLen desc
limit 20
\end{verbatim}


In Query 4, we cluster the drives by their common components. It clusters $166556$ from $224265$ drives into $7063$ clusters. There is a big cluster with $67151$ elements which is a single component (a cycle). This cycle is a trivial cycle in which the application starts per default. In this clustering we are considering the exact common components. On the other hand, if we use only cycles for clustering we will have $180999$ drives clustered into $6048$ clusters. To cluster the remaining singleton clusters, we can use other scoring possibilities such as Jaccard distance, to compare the common neighborhood between drives. 
Here is the Query 4 in Model variant 2:
\begin{verbatim}
MATCH (n:Drive)-[:Has]->(:State)-[c]->(:State) 
with n, c.compHash as compHash
order by compHash
with n, collect(distinct compHash) as cluster 
with collect(n.name) as Drives, cluster 
where size(Drives) > 1
return cluster, size(Drives) as clusterSize 
order by size(Drives) desc
\end{verbatim}
The same query in Model variant 3 looks like this:
\begin{verbatim}
MATCH (n:Drive)-[:Has]->(c:Circles) 
with n, c.name as compHash
order by compHash
with n, collect(distinct compHash) as Cluster 
with collect(n.name) as Drives, Cluster 
where size(Drives) > 1
return Cluster, size(Drives) as clusterSize 
order by size(Drives) desc
\end{verbatim}

One of the major challenges however remains the design of the proper injection pipeline for the components into the graph database. This comes from the fact that we have to avoid the duplicated components. That means at each insertion we have to check whether the component already exists in the database. We propose a hash indexing of the components in a look-up table to address this issue.

\section{Discussion}
In this work we have discussed the idea of remodeling the data representation and storage system, which can provide new possibilities for data analysis. In the case of sequence analysis, instead of traditional vectorization methods, we suggest a graph component-wise analysis. The concept behind it is derived from the fact that sequence itself is a traversal of a finite state automata. Based on this assumption, we introduce a new way of reviewing a sequence and consideration of loops. The main hypothesis is that the variation of components in a real use case converts. This assumption is very important because in a fully connected graph the number of possible simple paths and cycles increases exponentially with the number of nodes and edges. The number of simple paths in a graph with $n$ nodes for example can be approximated by the size of possible subsets of a $n$-element set, $2^n$. Thus, we have to examine the hypothesis on a larger amount of data. 

In general, due to the huge amount of data produced by the vehicles, the scalability is of great interest to us. In analysis of the model variants from section~\ref{models}, we evaluated a specific implementation of Neo4j and highlighted its limitations. Further investigation of other graph databases and their comparative study is therefore necessary.

Last but not least, we introduce a platform for ad hoc analysis of customer click data. Most of the relational databases are equipped with dashboards and graphical visualization interfaces which make it easier for the end-user to use those systems. Graph databases have their own query languages which is not familiar to most end-users.Therefore, an evaluation of user friendly interfaces and visualization tools for graph databases is essential.

\medskip

\bibliographystyle{unsrt}
\bibliography{bib}

\end{document}